\documentclass[12pt]{article}

\usepackage{amssymb}
\usepackage{amsthm,amsmath}
\usepackage{graphics,verbatim,color}
\usepackage[pdftex]{graphicx}
\usepackage{epstopdf}
\usepackage[ruled,lined,linesnumbered]{algorithm2e}
\usepackage{threeparttable}
\usepackage{epsfig,fancyhdr,times,hyperref}

\newtheorem{Theor}{Theorem}[section]
\newtheorem{Lemma}{Lemma}[section]

\newtheorem{Propo}{Proposition}[section]
\newtheorem{Remar}{Remark}[section]
\newtheorem{Defin}{Definition}[section]

\newtheorem{Examp}{Example}[section]

\def\recount{{\setcounter{equation}{0}}}
\def\recoup{{\setcounter{subsection}{0}}}

\begin{document}
\title{An application of the method of moments to volatility estimation 
using daily high, low, opening and closing prices
}


\maketitle





\bigskip

\begin{center}
\noindent{\bf Cristin Buescu}

Department of Mathematics, King's College, London

\medskip

\noindent{\bf Michael Taksar}\footnote{This work was supported by 
the Norwegian Research Council: Forskerprosjekt ES445026, ``Stochastic
Dynamics of Financial Markets."}\\
Mathematics Department, University of Missouri

\medskip

\noindent{\bf  Fatoumata J. Kon\'{e}}\\
Citibank, London

\end{center}





\medskip

\begin{center}
 {\bf Abstract}
\end{center}

We use the 
expectation of the range of an arithmetic Brownian
motion and the method of moments on the daily high, low, opening and closing
prices to estimate the 
volatility of the stock price. 
The daily price jump at the opening is 
considered to be the result of the unobserved evolution of an after-hours 
virtual trading day.
The annualized volatility is used to calculate Black-Scholes prices for European options,
and a trading strategy is devised to profit when these prices differ flagrantly from
the market prices.

\bigskip


%

\bigskip



%
%

%
%






%
%

%
%
%



\noindent

\renewcommand{\theequation}{\arabic{section}.\arabic{equation}}


\section{Introduction}
\recount
\recoup
\setcounter{equation}{0}


This article is a modified version of what has been studied in the Ph.D. thesis
of Kon\'{e} (1996). It concerns the application of the method of moments to 
range-based volatility estimation using daily high, low, opening and closing 
stock prices. 
Aiming to estimate volatility and not to measure it, we assume a Black-Scholes
framework with constant volatility and use daily data to achieve it (see 
Rogers and Zhou (2008) for further motivation for this choice).

Incidental to this is the derivation of the density and expectation of the
range of an arithmetic Brownian motion. Subsequent to this thesis, portions of 
it have been studied for different purposes (see, for instance, Sutrick et al
(1997) for the use of the density of the range of an arithmetic Brownian motion
in the {\em do-nothing option}, or Magdon-Ismail et al (2000, 2004) for the use of 
the expectation of the range of an arithmetic Brownian motion in different
contexts). In particular, expressing the density of the range in the context
of Sutrick et al (1997) corrects their expression.

The literature on range-based volatility estimation includes classic work by
Garman and Klass (1980), Parkinson (1980), Rogers and Satchell (1991) and 
Rogers et al (1994), whose estimators are reviewed in Yang and Zhang (2000). 
Of these, the latter paper is most related to the current one because it
considers after-hours price jumps in addition to drift. However, the methods 
presented here are different\footnote{see
Remarks \ref{important} and \ref{comparison}}, 
and perhaps more practical (see the remarks of 
Chan and Lien (2003) on the empirical availability of certain parameters in 
Yang and Zhang (2000)).

Starting from the joint density of the running maximum and the current value of
an arithmetic Brownian motion, the density of their difference (referred to as
half-range) is obtained, and its expectation computed. This allows the 
computation of the expectation of the full range (defined as maximum minus 
minimum), which will then be used in the method of moments for intra-day
volatility estimation.

After-hours arrival of information results in price jumps at the opening, and
we model this as a virtual trading day which is unobservable, but which,
when succeeding the trading day, gives on average the complete statistical
representation of one day.

Black-Scholes option prices are computed using the parameter estimates, and
when they differ the most from the observed market prices a profit is made by
an appropriately devised trading strategy.


The paper is organized as follows. In Section 2 we derive the expectation and
the density function of an arithmetic Brownian motion. In Section 3 the
method of moments is used to estimate the parameters of the stock price based on
daily high, low, opening and closing data. The estimated parameters are then used in Section 4
to price European options on the stock, which are then compared to market
prices to identify instances of flagrant differences.
The effect of the mispricing is estimated by computing the 
profit to be made in these opportunities. We conclude in Section 5 with 
some comments on the efficiency of the method of moments in volatility 
estimation.

\section{The range of an arithmetic Brownian motion: expectation and density }
\recount
\recoup
Let $\{\Omega,{\cal F},P\}$ be a probability space 
endowed with a filtration $\{{\cal F}_{t}\}_{t\geq 0}$,
and let $\{W_{t}\}_{t\geq 0}$ be a one-dimensional standard Brownian motion
adapted to $\{{\cal F}_{t}\}_{t\geq 0}$. 
For $t\geq 0$ let $X_t$ denote a standard arithmetic Brownian motion with 
drift $\mu$ and volatility $\sigma>0$:
\begin{equation}
X_t =  \mu \;t +\sigma \;W_t ,\;\;X_0=0,
          \label{equation1}
\end{equation}
and let $\{M_{t}\}_{t\geq 0}$, $\{m_{t}\}_{t\geq 0}$ and $\{R_{t}\}_{t\geq 0}$
denote its running maximum, minimum, and range, respectively:
\begin{equation}
M_t:=\sup_{0\leq s\leq t}{X_s},\;\;m_t=\inf_{0\leq s\leq t}{X_s},\;\;
R_t:=M_t-m_t.
\label{equation2}
\end{equation}

\medskip

First we derive the expectation $E[R_t]$ of the range of the arithmetic Brownian
motion $X_t$. This is achieved by computing the density and expectation of the
half-range $M_t-X_t$ from the joint density of $X_t$ and $M_t$.
 
\begin{Lemma} \label{lemma1}
The joint density function of an arithmetic Brownian motion and its
running maximum can be expressed as:
\begin{equation}
P(X_t\in da,M_t\in db)=\frac{2(2b-a)}{\sqrt{2\pi t^3}\sigma^3}\exp{
\left\{-\frac{(2b-a)^2}{2t\sigma^2}+\frac{\mu}{\sigma^2}a-\frac 12
\frac{\mu^2}{\sigma^2}t\right\}}\;da\;db. \label{jointmM}
\end{equation}
\end{Lemma}

\begin{proof} The proof is standard. 
Using the martingale $Z_t(Y)=\exp\{W_t\mu/\sigma-t\mu^2/(2\sigma^2)\}$, 
Girsanov's change of measure defines a new probability measure $\tilde{P}$ for
any measurable set $A$ by $\tilde{P}(A)=E[Z_t(Y)1_A]$.
Theorem 3.2.2 of Karatzas and Shreve (1988) with $Y_t=\mu/\sigma$ and
$N_t=\sigma W_t$ gives that $\tilde{N}_t=\sigma W_t-\mu t$ is a local
martingale.
The process $\tilde{W}_t=W_t-t\mu/\sigma$ is 
a Brownian motion under the new probability measure $\tilde{P}$. Equivalently, 
$\sigma W_t=\mu t+\sigma \tilde{W}_t$ is a Brownian
motion with drift under $\tilde{P}$, and we can write:
\begin{eqnarray}
P(X_t\leq a, {M}_t\leq b)&=&\tilde{P}(\sigma W_t\leq a, \sigma \bar{M}_t\leq b)
\nonumber\\
&=&\int_{-\infty}^{a} \exp\Big{(}\frac{\mu}{\sigma^2}x
-\frac 12 \frac{\mu^2}{\sigma^2}t\Big{)}
\;P(\sigma W_t\in dx, \sigma \bar{M}_t\leq b).\label{jtildeP}
\end{eqnarray}

An application of the reflection principle gives:
\begin{eqnarray*}
P(\sigma W_t\leq a, \sigma \bar{M}_t\leq b)&=&
P\Big{(}W_t\leq \frac{a}{\sigma}\Big{)}
-P\Big{(}W_t\leq \frac{a}{\sigma},\sigma \bar{M}_t> b\Big{)}\\
&=&P\Big{(}W_t\leq \frac{a}{\sigma}\Big{)}-P\Big{(}W_t>\frac{2b-a}{\sigma}\Big{)}
\\
&=&\Phi\Big{(}\frac{a}{\sigma\sqrt{t}}\Big{)}-1+
\Phi\Big{(}\frac{2b-a}{\sigma\sqrt{t}}\Big{)},
\end{eqnarray*}
where we denote by $\phi(\cdot)$ and $\Phi(\cdot)$ the standard normal density 
and cumulative distribution functions, respectively. 

Differentiating the formula above with respect to $a$ gives:
\begin{equation}
P(\sigma W_t\in da, \sigma \bar{M}_t\leq b )=\frac{1}{\sigma \sqrt{t}}
\Bigg{(}\phi\Big{(}\frac{a}{\sigma\sqrt{t}}\Big{)}-
\phi\Big{(}\frac{2b-a}{\sigma\sqrt{t}}\Big{)}\Bigg{)}.\label{jP}
\end{equation}

Replacing (\ref{jP}) in (\ref{jtildeP}) and differentiating first with respect to 
$a$ gives:
\begin{eqnarray*}
P(X_t\in da, M_t\leq b)=\frac{1}{\sigma\sqrt{t}}
\exp\Big{(}\frac{\mu}{\sigma^2}a-\frac 12 \frac{\mu^2}{\sigma^2}t\Big{)}
\Bigg{(}\phi\Big{(}\frac{a}{\sigma\sqrt{t}}\Big{)}-
\phi\Big{(}\frac{2b-a}{\sigma\sqrt{t}}\Big{)}\Bigg{)}da,
\end{eqnarray*}
and differentiating then with respect to $b$ gives:
\begin{eqnarray*}
P(X_t\in da, M_t \in db)=\frac{1}{\sigma\sqrt{t}}
\exp\Big{(}\frac{\mu}{\sigma^2}a-\frac 12 \frac{\mu^2}{\sigma^2}t\Big{)}
\Big{(}-\frac{2}{\sigma\sqrt{t}}\Big{)}\phi'\Big{(}\frac{2b-a}{\sigma\sqrt{t}}
\Big{)}\;da\;db\\
=\Big{(}-\frac{2}{t\sigma^2}\Big{)}\exp\Big{(}\frac{\mu}{\sigma^2}a-
\frac 12 \frac{\mu^2}{\sigma^2}t\Big{)}\frac{1}{\sqrt{2\pi}}
\exp\Big{(}-\frac 12\Big{(}\frac{2b-a}{\sigma\sqrt{t}}\Big{)}^2\Big{)}
\Big{(}-\frac 12\Big{)}2\frac{2b-a}{\sigma\sqrt{t}}\;da\;db\\
=\frac{2(2b-a)}{\sqrt{2\pi t^3}\sigma^3}\exp{
\left\{-\frac{(2b-a)^2}{2t\sigma^2}+\frac{\mu}{\sigma^2}a-\frac 12
\frac{\mu^2}{\sigma^2}t\right\}}\;da\;db.
\end{eqnarray*}
The joint density $f_{X_t,M_t}(a,b)$ is given by the term multiplying $da\;db$ 
above.
\end{proof}

\begin{Remar}
This is one of those results that seemed to be always at hand (it can be
obtained from equation (1.8.8) of Harrison (1985)), but never derived.
Note the typo in Yang and Zhang (2000), whose expression (B1) has a plus for the
first fraction in the exponential.  For $\sigma=1$ this was used in Example E5
of Karatzas and Shreve (1998) in relationship to Clark's formula to obtain
explicitly the hedging portfolio.
\end{Remar}

\medskip

The density of the half-range $M_t-X_t$ can be obtained using a standard
two-dimensional transformation of the above joint density.

\begin{Lemma} \label{lemma2}
The density of the half-range $M_t-X_t$ is given by:
\begin{eqnarray}
f_{M_t-X_t}(c)=2\frac{\mu}{\sigma^2}\;\Phi\Big{(}
\frac{\mu t-c}{\sigma\sqrt{t}}\Big{)}
\exp\Big{(}-2\frac{\mu}{\sigma^2}c\Big{)}
+\frac{2}{\sigma\sqrt{2t\pi}}
\exp\left\{-\frac{(\mu t +c)^2}{2t\sigma^2}\right\}.
\end{eqnarray}
\end{Lemma}
\begin{proof}
For $Y_1=X_1+X_2$ and $Y_2=X_1-X_2$ 
the joint density is
$
f_{Y_1,Y_2}(y_1,y_2)=\frac 12 f_{X_1,X_2}\Big{(}\frac{y_1+y_2}{2},\frac{y_1-y_2}{2}\Big{)}.
$
Taking $X_1=M_t$ and $X_2=X_t$ and using Lemma \ref{lemma1} gives: 
\begin{eqnarray*}
f_{Y_1,Y_2}(y_1,y_2)=\frac 12 f_{M_t,X_t}\Big{(}\frac{y_1+y_2}{2},\frac{y_1-y_2}{2}\Big{)}
=\frac 12 \;f_{X_t,M_t}\Big{(}\frac{y_1-y_2}{2},\frac{y_1+y_2}{2}\Big{)}\\
=\frac 12 \;\frac{2y_1+2y_2-y_1+y_2}{\sqrt{2\pi t^3}\sigma^3}
\exp{\left\{-\frac{(2y_1+2y_2-y_1+y_2)^2}{8t\sigma^2}+\frac{\mu}{\sigma^2}
\frac{y_1-y_2}{2}-\frac 12 \frac{\mu^2}{\sigma^2}t\right\}}\\
=\frac 12 \;\frac{y_1+3y_2}{\sqrt{2\pi t^3}\sigma^3}
\exp{\left\{-\frac{(y_1+3y_2)^2}{8t\sigma^2}+\frac{\mu}{\sigma^2}
\frac{y_1-y_2}{2}-\frac 12 \frac{\mu^2}{\sigma^2}t\right\}}\\
=\frac 12 \;\frac{y_1+3y_2}{\sqrt{2\pi t^3}\sigma^3}
\exp{\left\{-\frac 12 \frac{(y_1+3y_2-2\mu t)^2}{4t\sigma^2}
-2 \frac{\mu}{\sigma^2}y_2\right\}}
\end{eqnarray*}
Note that $M_t\geq 0$ implies 
$Y_1\geq -Y_2$, thus the marginal density of $Y_2=M_t-X_t$ is:
\begin{eqnarray*}
f_{Y_2}(y_2)=\int_{-y_2}^{\infty}f_{Y_1,Y_2}(y_1,y_2)dy_1.
\end{eqnarray*}

A change of variable $z=(y_1+3y_2-2\mu t)/(2\sigma\sqrt{t})$ gives: 
\[
z>z_0:=\frac{y_2-\mu t}{\sigma\sqrt{t}},\;\;
dy_1=2\sigma\sqrt{t}dz,
\]
therefore:
\begin{eqnarray*}
&&f_{Y_2}(y_2)=\int_{z_0}^{\infty}\frac 12 \frac{2z\sigma\sqrt{t}+2\mu t}
{\sqrt{2\pi t^3}\sigma^3}\exp{\left\{-\frac 12 {z^2}
-2 \frac{\mu}{\sigma^2}y_2\right\}}2\sigma\sqrt{t}dz\\
&=&\frac{2}{\sigma\sqrt{t}}\int_{z_0}^{\infty}z\frac{1}{\sqrt{2\pi}}e^{-\frac{z^2}{2}}dz
\exp\left\{-2\frac{\mu}{\sigma^2}y_2\right\}+2\frac{\mu}{\sigma^2}\int_{z_0}^\infty 
\frac{1}{\sqrt{2\pi}}e^{-\frac{z^2}{2}}dz\exp\left\{-2\frac{\mu}{\sigma^2}y_2\right\}\\
&=&\frac{2}{\sigma\sqrt{t}}\int_{z_0}^{\infty}-\Big{(}\frac{1}{\sqrt{2\pi}}e^{-\frac{z^2}{2}}\Big{)}'dz
\exp\left\{-2\frac{\mu}{\sigma^2}y_2\right\}+2\frac{\mu}{\sigma^2}\Big{(}1-\Phi(z_0)\Big{)}
\exp\left\{-2\frac{\mu}{\sigma^2}y_2\right\}\\
&=&\frac{2}{\sigma\sqrt{t}}\frac{1}{\sqrt{2\pi}}e^{-\frac{z_0^2}{2}}
\exp\left\{-2\frac{\mu}{\sigma^2}y_2\right\}+2\frac{\mu}{\sigma^2}\Phi(-z_0)
\exp\left\{-2\frac{\mu}{\sigma^2}y_2\right\}.
\end{eqnarray*}
This can be rewritten as:
\begin{eqnarray}
P(M_t-X_t\in dc)&=&\Bigg{(}2\frac{\mu}{\sigma^2}\Phi\Big{(}
\frac{\mu t-c}{\sigma\sqrt{t}}\Big{)}
\exp\Big{(}-2\frac{\mu}{\sigma^2}c\Big{)}\nonumber\\
&+&\frac{2}{\sigma\sqrt{2t\pi}}
\exp\left\{-\frac{(\mu t -c)^2}{2t\sigma^2}\right\}
\exp\Big{(}-2\frac{\mu}{\sigma^2}c\Big{)}
\Bigg{)}dc, \label{halfdensity}
\end{eqnarray}
and the result follows.
\end{proof}

\begin{Propo} \label{propo1} 
The expectation of the half-range is given by:
\begin{eqnarray}
E(M_t-X_t)&=&\frac{\sigma^2}{2\mu}\;\Phi\Big{(}\frac{\mu}{\sigma} \sqrt{t}\Big{)}
-\Big{(}\mu t+\frac{\sigma^2}{2\mu} \Big{)}\Big{(}1-\Phi\Big{(}\frac \mu \sigma
 \sqrt(t)\Big{)}\Big{)}\nonumber\\
&+&\frac{\sigma\sqrt{t}}{\sqrt{2\pi}}
\exp{\Big{(}-\frac{t \mu^2}{2\sigma^2}\Big{)}}.  \label{halfexpectation}
\end{eqnarray}
\end{Propo}

\begin{proof} 
A simple calculation yields:
\begin{eqnarray*}
E(M_t-X_t)&=&\int_0^\infty c\;\frac{2\mu}{\sigma^2} \;\Phi\Big{(}
\frac{\mu t-c}{\sigma\sqrt{t}}\Big{)}\exp\Big{(}-2\frac{\mu}{\sigma^2}c\Big{)}
dc\\
&&+\int_0^\infty c \frac{2}{\sigma\sqrt{2t\pi}}
\exp\left\{-\frac{(\mu t +c)^2}{2t\sigma^2}\right\}dc\\
&=&\int_0^\infty 
\left\{
-(c+\frac{\sigma^2}{2\mu})
\exp\Big{(}-2\frac{\mu}{\sigma^2}c\Big{)}
\right\}'
\;\Phi\Big{(}\frac{\mu t-c}{\sigma\sqrt{t}}\Big{)}dc\\
&&+\int_0^\infty c \frac{2}{\sigma\sqrt{2t\pi}}
\exp\left\{-\frac{(\mu t +c)^2}{2t\sigma^2}\right\}dc\\
=\frac{\sigma^2}{2\mu}\;\Phi\Big{(}\frac\mu\sigma\sqrt{t}\Big{)}
&+&\frac{1}{\sigma\sqrt{t}}\frac{1}{\sqrt{2\pi}}\int_0^\infty 
\Big{(}-c-\frac{2\sigma^2}{\mu}+2c\Big{)}
\exp\left\{-\frac{(\mu t +c)^2}{2t\sigma^2}\right\}dc.
\end{eqnarray*}
A change of variable $c=z\sigma\sqrt{t}-\mu t$ gives the result.
\end{proof}

\medskip

Consider now $E[X_t-m_t]$. For each path of the Brownian motion $X_t$ 
with drift $\mu$ consider a symmetric path of a Brownian motion $\tilde X_t$ 
having drift $-\mu$. Then $X_t-m_t=-(\tilde X_t-\tilde M_t)$ and $E[X_t-m_t]$ 
can be calculated using the equation (\ref{halfexpectation}) with $\mu$ 
replaced by $-\mu$. Whereas the formula for the expectation of the range 
follows.

\begin{Theor}\label{theor1}
The expectation of the range of the arithmetic Brownian motion $X_t$ defined in
(\ref{equation1}) is given by:
\begin{equation}
E[R_t]=\Big{(}\mu t+\frac{\sigma^2}{\mu}\Big{)}
\Bigg{(}1-2\Phi\Big{(}-\sqrt{t}\frac{\mu}{\sigma}\Big{)}\Bigg{)}
+2\frac{\sigma\sqrt{t}}{\sqrt{2\pi}}
\exp\Big{(}-\frac{t\mu^2}{2\sigma^2}\Big{)}. \label{equation5}
\end{equation}
\end{Theor}
Let us denote this
expected range function by $ER(\mu,\sigma,t)$. On closer inspection this can
be further simplified as a function of just two quantities:
\begin{equation}
E[R_t]=ER(\mu,\sigma,t)=h\Big{(}\frac{\mu t}{\sigma \sqrt{t}},
\frac{\sigma^2 }{\mu }\Big{)},\label{er}
\end{equation}
where the function $h$ is defined by:
\begin{equation}
h(x,y):=\left\{(x^2+1)(2\Phi(x)-1)+\frac{2x}{\sqrt{2\pi}}
\exp\Big{(}-\frac{x^2}{2}\Big{)}\right\}y. \label{er2}
\end{equation}
Note that $ER(\mu,\sigma,t)=ER(\mu t,\sigma\sqrt{t},1)$ as it should (the range
over a time interval $(0,t)$ of an arithmetic Brownian motion with parameters
$\mu$ and $\sigma$ is the same as that over $(0,1)$ when the parameters change
to $\mu t$ and $\sigma\sqrt{t}$).

\bigskip

In the remainder of this section we derive the density of the range $R_t$ of
the arithmetic Brownian motion $X_t$. This is achieved by the use of the joint
density of the minimum and the maximum of $X_t$, a result with its own merit,
that we could not find published prior to Kon\'{e} (1996) (Borodin and Salminen
(1996, 1.15.4) published in the same year the joint cumulative distribution 
function only in terms of some definite integrals).

A version of this result is used in Sutrick et al (1997) for the same purpose,
but there it seems to have incorporated an error.

\medskip




To obtain the joint density $F(a,b)$ of the maximum and the minimum we start 
with a lemma.
\begin{Lemma}\label{lemma_main}
We can write:
\begin{eqnarray}
F(a,b)=\int_a^bh(a,b,x)\exp{\Big{(}\frac{\mu}{\sigma^2}x-\frac 12
\frac{\mu^2}{\sigma^2}t\Big{)}}\;dx,\label{aux28}
\end{eqnarray}
where
\begin{eqnarray}
h(a,b,x)&=&h_1(a,b,x)-h_2(a,b,x),\\
h_1(a,b,x)&=&\sum_{k=-\infty}^\infty\frac{2k(2k-2)}{\sigma^3 t\sqrt{2\pi t}}
\Bigg{[}
 1-\frac{[2k(b-a)-2b+x]^2}{t\sigma^2}
\Bigg{]}\nonumber\\
&&\times\exp{\Big{(}-\frac{[2k(b-a)-2b+x]^2}{2t\sigma^2}\Big{)}},
\nonumber\\
h_2(a,b,x)&=&\sum_{k=-\infty}^\infty\frac{4k^2}{\sigma^3 t\sqrt{2\pi t}}
\Bigg{[}
 1-\frac{[2k(b-a)-x]^2}{t\sigma^2}
\Bigg{]}
\exp{\Big{(}-\frac{[2k(b-a)-x]^2}{2t\sigma^2}\Big{)}}.\nonumber
\end{eqnarray}
\end{Lemma}

\begin{proof}
Using the change of measure of the proof of Lemma \ref{lemma1} and 
Girsanov's theorem we have:
\begin{equation}
P(a<m_t<M_t<b)=\int_a^b p_t(x;a,b)\exp{\Big{(}\frac{\mu}{\sigma^2}x-\frac 12
\frac{\mu^2}{\sigma^2}t\Big{)}}dx,\label{aux21}
\end{equation}
which gives $F(a,b)$ via:
\begin{equation}
F(a,b)=-\frac{\partial^2 P(a<m_t<M_t<b)}{\partial a \;\partial b}.\label{aux22}
\end{equation}
Leibniz rule of differentiation:
\begin{equation}
\frac{\partial}{\partial z}\int_{a(z)}^{b(z)}f(x,z)\;dx=
\int_{a(z)}^{b(z)}\frac{\partial f(x,z)}{\partial z}dx
+f(b(z),z)\frac{\partial b}{\partial z}-f(a(z),z)
\frac{\partial a}{\partial z}
\label{Leibniz}
\end{equation}
gives the partial derivative of (\ref{Feller}) wrt $b$:
\begin{eqnarray}
\int_a^b\frac{1}{\sigma\sqrt{t}}
\sum_{k=-\infty}^{\infty}\Bigg{[}\frac{\partial\phi}{\partial b}
\Big{(}\frac{2k(b-a)-x}{\sigma\sqrt{t}}\Big{)}
-\frac{\partial\phi}{\partial b}
\Big{(}\frac{2k(b-a)-2b+x}{\sigma\sqrt{t}}\Big{)}
\Bigg{]}
\nonumber\\
\times \exp{\Big{(}\frac{\mu}{\sigma^2}x-\frac 12
\frac{\mu^2}{\sigma^2}t\Big{)}}dx.\label{aux24}
\end{eqnarray}
Differentiating wrt $a$ this last equation gives:
\begin{eqnarray}
F(a,b)=-\int_a^b\frac{1}{\sigma\sqrt{t}}
\sum_{k=-\infty}^{\infty}\Bigg{[}\frac{\partial^2\phi}{\partial a\;\partial b}
\Big{(}\frac{2k(b-a)-x}{\sigma\sqrt{t}}\Big{)}
-\frac{\partial^2\phi}{\partial a\;\partial b}
\Big{(}\frac{2k(b-a)-2b+x}{\sigma\sqrt{t}}\Big{)}
\Bigg{]}
\nonumber\\
\times 
\exp{\Big{(}\frac{\mu}{\sigma^2}x-\frac 12
\frac{\mu^2}{\sigma^2}t\Big{)}}\;dx.\hspace*{0.4cm}\label{aux25}
\end{eqnarray}
Direct calculation gives:
\begin{eqnarray}
\frac{\partial^2\phi}{\partial a\;\partial b}
\Big{(}
  \frac{2k(b-a)-x}{\sigma\sqrt{t}}
\Big{)}
=
\frac{\partial}{\partial a}
\Bigg{[}
  \frac{(-2k)[2k(b-a)-x]}{\sqrt{2\pi}t\sigma^2} 
  \exp{\Big{(}-\frac{[2k(b-a)-x]^2}{2t\sigma^2}\Big{)}}
\Bigg{]}
\nonumber\\
=
\frac{4k^2}{\sqrt{2\pi}t\sigma^2}
\exp{\Big{(}-\frac{[2k(b-a)-x]^2}{2t\sigma^2}\Big{)}}
\Big{(}
  1-\frac{[2k(b-a)-x]^2}{t\sigma^2}
\Big{)},\hspace*{0.3cm}\label{aux26}
\end{eqnarray}
\begin{eqnarray}
\frac{\partial^2\phi}{\partial a\;\partial b}
\Big{(}
 \frac{2k(b-a)-2b+x}{\sigma\sqrt{t}}
\Big{)}\hspace*{9cm}
\nonumber\\
=
\frac{\partial}{\partial a}
\Bigg{[}
 \frac{-(2k-2)[2k(b-a)-2b+x]}{\sqrt{2\pi}t\sigma^2} 
\exp{\Big{(}-\frac{[2k(b-a)-2b+x]^2}{2t\sigma^2}\Big{)}} 
\Bigg{]}\hspace*{1cm}\label{aux27}
\nonumber\\
=
\frac{4k(k-1)}{\sqrt{2\pi}t\sigma^2}
\exp{\Big{(}-\frac{[2k(b-a)-2b+x]^2}{2t\sigma^2}\Big{)}} 
\Big{(}
 1-\frac{[2k(b-a)-2b+x]^2}{t\sigma^2}
\Big{)}.\hspace*{0.7cm}
\end{eqnarray}
Substituting (\ref{aux26})-(\ref{aux27}) in (\ref{aux25}) gives
the result.
\end{proof}

\begin{Propo} \label{propo_joint}
The joint density function $F(a,b)$ of $M_t$ and $m_t$ can be represented as:
\begin{eqnarray}
F(a,b)&=&F_1(a,b)-F_2(a,b)-F_3(a,b)+F_4(a,b)\\
&-&F_5(a,b)+F_6(a,b)+F_7(a,b)-F_8(a,b),\nonumber
\end{eqnarray}
with
\begin{eqnarray}
\lefteqn{F_1(a,b)=\sum_{k=-\infty}^{\infty}\frac{4k(k-1)}{t\sigma^3\sqrt{2\pi t}}
\;\;[(2k-1)b-2ka+\mu t]}\hspace*{2cm}\nonumber
\\
&&\times\exp\left\{
-\frac{\mu}{\sigma^2}[2(k-1)b-2ka]-\frac{[(2k-1)b-2ka-\mu t]^2}{2t\sigma^2}
\right\},\hspace*{0.8cm}
\end{eqnarray}
\begin{eqnarray}
\lefteqn{F_2(a,b)=\sum_{k=-\infty}^{\infty}\frac{4k(k-1)}{t\sigma^3\sqrt{2\pi t}}
\;\;[2(k-1)b-(2k-1)a+\mu t]}\hspace*{1cm}\nonumber
\\
&&\times\exp\left\{
-\frac{\mu}{\sigma^2}[2(k-1)b-2ka]-\frac{[2(k-1)b-(2k-1)a-\mu t]^2}{2t\sigma^2}
\right\},\hspace*{0.9cm}
\end{eqnarray}
\begin{eqnarray}
F_3(a,b)=\sum_{k=-\infty}^{\infty}\frac{4k(k-1)\mu^2}{2\sigma^4}
\exp\left\{
-\frac{\mu}{\sigma^2}[2(k-1)b-2ka]
\right\}\hspace*{3cm}\nonumber
\\ 
\times\operatorname{erf}\Bigg{(}\frac{(2k-1)b-2ka-\mu t}{\sigma\sqrt{2t}}\Bigg{)},
\hspace*{0.6cm}
\end{eqnarray}
\begin{eqnarray}
F_4(a,b)=\sum_{k=-\infty}^{\infty}\frac{4k(k-1)\mu^2}{2\sigma^4}
\exp\left\{
-\frac{\mu}{\sigma^2}[2(k-1)b-2ka]
\right\}\hspace*{3cm}\nonumber 
\\
\times\operatorname{erf}\Bigg{(}\frac{2(k-1)b-(2k-1)a-\mu t}{\sigma\sqrt{2t}}\Bigg{)},
\hspace*{0.6cm}
\end{eqnarray}
\begin{eqnarray}
\lefteqn{F_5(a,b)
=\sum_{k=-\infty}^{\infty}\frac{4k^2}{t\sigma^3\sqrt{2\pi t}}
\;\;[(2k+1)b-2ka+\mu t]}\hspace*{4cm}\nonumber
\\
&&\times\exp\left\{
-\frac{\mu}{\sigma^2}2k(b-a)-\frac{[(2k+1)b-2ka-\mu t]^2}{2t\sigma^2}
\right\},\hspace*{1cm}
\end{eqnarray}
\begin{eqnarray}
\lefteqn{F_6(a,b)
=\sum_{k=-\infty}^{\infty}\frac{4k^2}{t\sigma^3\sqrt{2\pi t}}
\;\;[2kb-(2k-1)a+\mu t]}\hspace*{5cm}
\\
&&\times\exp\left\{
-\frac{\mu}{\sigma^2}2k(b-a)-\frac{[2kb-(2k-1)a-\mu t]^2}{2t\sigma^2}
\right\},\hspace*{0.3cm}\nonumber
\end{eqnarray}
\begin{eqnarray}
F_7(a,b)=\sum_{k=-\infty}^{\infty}\frac{4k^2\mu^2}{2\sigma^4}
\exp\left\{
-\frac{\mu}{\sigma^2}2k(b-a)
\right\}\; 
\operatorname{erf}\Bigg{(}\frac{(2k+1)b-2ka-\mu t}{\sigma\sqrt{2t}}\Bigg{)},
\end{eqnarray}
\begin{eqnarray}
F_8(a,b)=\sum_{k=-\infty}^{\infty}\frac{4k^2\mu^2}{2\sigma^4}
\exp\left\{
-\frac{\mu}{\sigma^2}2k(b-a)
\right\}\; 
\operatorname{erf}\Bigg{(}\frac{2kb-(2k-1)a-\mu t}{\sigma\sqrt{2t}}\Bigg{)},
\end{eqnarray}
where
\begin{equation}
\operatorname{erf}(x)=\frac{2}{\sqrt{\pi}}\int_0^x e^{-t^2}dt.\label{aux33}
\end{equation}
\end{Propo}

\begin{proof}
From Lemma \ref{lemma_main} we write $F(a,b)$ as the difference of two terms,
which we denote $I_1$ and $I_2$:

\begin{eqnarray}
I_1:=\int_a^b h_1(a,b,x)\exp{\Big{(}\frac{\mu}{\sigma^2}x-\frac 12
\frac{\mu^2}{\sigma^2}t\Big{)}}\;dx,\nonumber\\
I_2:=\int_a^b h_2(a,b,x)\exp{\Big{(}\frac{\mu}{\sigma^2}x-\frac 12
\frac{\mu^2}{\sigma^2}t\Big{)}}\;dx.\label{i1}
\end{eqnarray}
In each $I_1$ and $I_2$ we combine the exponents and then use, respectively, 
a change of variable:      
\begin{equation}
z=\frac{x+2k(b-a)-2b-\mu t}{\sigma\sqrt{t}},\;\;
z=\frac{x+2k(b-a)-\mu t}{\sigma\sqrt{t}}, \label{change12}
\end{equation}
followed by an integration by parts for $\int z^2\exp(-z^2/2)dz$ and replacement
of $\Phi(x)$ by $\operatorname{erf}(x)$ via $\Phi(x)=0.5\operatorname{erf}
(x/\sqrt{2})+0.5$. The resulting eight terms are then then denoted $F_i(a,b)$,
$i=1,\dots, 8$.
\end{proof}
\medskip

We use the above expression to derive the density of the range. To make it
suitable for comparison with that obtained by Sutrick et al (1997) in their 
Proposition 1, we change $k\rightarrow k+1$ in the summations of 
$F_1, F_2, F_3$ and $F_4$.

\begin{Propo}
The density function $f_{R_t}(r)$ for the range of an arithmetic Brownian motion 
can be written as:
\begin{equation}
f_{R_t}(r)=\frac{1}{\sigma\sqrt{t}}\sum_{k=-\infty}^{\infty} 4k^2 I(k)
+\frac{1}{\sigma\sqrt{t}}\sum_{k=-\infty}^{\infty}4k(k+1)J(k),
\end{equation}
where
\begin{eqnarray}
I(k)&=&e^{-\frac{2\mu kr}{\sigma^2}}(1+c^2)
(\phi(K_1-c)-2\phi(K_0-c)+\phi(K_{-1}-c))\nonumber\\
&+&e^{-\frac{2\mu kr}{\sigma^2}}
[(c^2K_1-2c-c^3)\Phi(K_1-c)-2(c^2K_0-2c-c^3)\Phi(K_0-c)\nonumber\\
&+&(c^2K_{-1}-2c-c^3)\Phi(K_{-1}-c)], 
\end{eqnarray}
and
\begin{eqnarray}
J(k)&=&e^{\frac{2\mu kr}{\sigma^2}}(\phi(K_1+c)-\phi(K_0+c))
-e^{\frac{2\mu(k+1)r}{\sigma^2}}(\phi(K_2+c)-\phi(K_1+c))\nonumber\\
&&+e^{-\frac{2\mu kr}{\sigma^2}}\Big{(}-\frac c2 \Phi(K_1-c)
+\frac c2 \Phi(K_0-c)\Big{)}\nonumber\\
&&-e^{-\frac{2\mu (k+1)r}{\sigma^2}}\Big{(}-\frac c2 \Phi(K_2-c)
+\frac c2 \Phi(K_1-c)\Big{)}
\nonumber\\
&&+e^{\frac{2\mu kr}{\sigma^2}}(\Phi(K_1+c)-\Phi(K_0+c))\nonumber\\
&&-e^{\frac{2\mu(k+1)r}{\sigma^2}}(\Phi(K_2+c)-\Phi(K_1+c)),
\end{eqnarray}
with 
\begin{eqnarray*}
K_2=\frac{(2k+2)r}{\sigma\sqrt{t}},\;K_1=\frac{(2k+1)r}{\sigma\sqrt{t}},\;
K_0=\frac{2kr}{\sigma\sqrt{t}},\;K_{-1}=\frac{(2k-1)r}{\sigma\sqrt{t}},\;
c=\frac{\mu\sqrt{t}}{\sigma}.
\end{eqnarray*}
\end{Propo}

\begin{proof}
After replacing $k$ by $k+1$ in $F_i$, ${i=1,\dots,4}$ of Proposition 
\ref{propo_joint},
a two-dimensional transformation $a=u-v, b=u$, gives, via Jacobian, the density
of the range and the running maximum. Its marginal density is the one we seek:
$$
f(r)=\int_0^rF(u-r,u)du.
$$
Applying a change of variable and integration by parts gives the result.
\end{proof}

\begin{Remar}\label{remark2}
This result 
corrects that of Sutrick et al (1997) where there appears to be a mistake
in the computations.
\end{Remar}

\begin{Remar}\label{equiv}
The probabilistic starting point for both Kon\'{e} (1996) and Sutrick et al 
(1997) is 
$p_t(x;a,b)$ $dx:=P(a<m_t<M_t<b,x\leq X_t<x+dx|X_0=0)$. The former uses a result 
of Feller (1951) that can be traced to L\'{e}vy (1948):
\begin{equation}
p_t(x;a,b)=\frac{1}{\sigma\sqrt{t}}\sum_{k=-\infty}^\infty
\Bigg{[}\phi\Big{(}\frac{2k(b-a)-x}{\sigma\sqrt{t}}\Big{)}
-\phi\Big{(}\frac{2k(b-a)-2b+x}{\sigma\sqrt{t}}\Big{)}\Bigg{]},
\label{Feller}
\end{equation}
 while the latter uses a result of Billingsley (1968):
\begin{equation}
p_t(x;a,b)=\frac{1}{\sigma\sqrt{t}}\sum_{k=-\infty}^\infty
\Bigg{[}\phi\Big{(}\frac{x+2k(b-a)}{\sigma\sqrt{t}}\Big{)}
-\phi\Big{(}\frac{2b-x+2k(b-a)}{\sigma\sqrt{t}}\Big{)}\Bigg{]}.
\label{Billingsley}
\end{equation}
The probabilistic results 
(\ref{Feller}) and (\ref{Billingsley}) are in fact equivalent, as one can be obtained from the other by 
appropriately replacing the summation index $k$ 
with 
$-k$ 
and $\phi(x)$ with $\phi(-x)$.
\end{Remar}
 
\section{The method of moments applied to volatility estimation
using  daily high, low, opening and closing prices}
\recount
\recoup
In this section we apply Theorem \ref{theor1} to the estimation of the drift 
and  volatility parameters of the stock price from 
market data on high, low, opening and closing prices.

\begin{Defin}\label{days}
i) A {\em trading day} is the period elapsed between the opening and
the closing bells of a calendar day.

ii) A {\em virtual trading day} is the after-hours period 
beginning from the closing of one trading day and ending at the opening 
of the next trading day.

iii) A {\em one-day period} consists of one trading day followed by one
virtual trading day.
\end{Defin}

We assume that the stock price $S_t$ 
has
the usual geometric Brownian motion dynamics:
\begin{equation}
\frac{dS_t}{S_t}=\mu_s\; dt+\sigma\; dW_t, \;\;t\geq 0.
\label{equation6} 
\end{equation}
Then the log-stock price $\log S_t$ is the arithmetic Brownian motion $X_t$ 
defined in (\ref{equation1})
with drift coefficient 
\begin{equation}
\mu=\mu_s-\frac{\sigma^2}{2}.   \label{equation7}
\end{equation} 
Note that $\mu_s$ is the {\em one-day period} drift of the stock price $S_t$, while 
$\mu$ is the similar drift of the log-price $X_t=\log S_t$;  the volatility 
parameter $\sigma$ is the same for both $S_t$ and $X_t$.

The market data used for parameter estimation is as follows: for each one-period
day $i\in\{1, 2, \dots, n\}$ we denote by $S_{i-1}$ the opening price and 
by 
$H_i$ and $L_i$ the intra-day high and low prices, respectively (i.e. the high and 
low are observed only during the trading day, and not the virtual trading day
- see Figure \ref{fig0}). 

The after-hours arrival of information in the market determines a jump between 
the closing price 
of one trading day and the opening price 
of the next day. We model this jump by letting the same geometric
Brownian motion $S_t$ have an unobserved evolution during a virtual trading
day. The length of this virtual trading day is assumed to be, on average,
a fraction $f$
of the unit length of the one-day period. 
\begin{Remar}\label{important}This assumption follows 
Garman and Klass (1980) and
Yang and Zhang (2000), except that they assume the after-hours trading day precedes
the actual trading day. They call it the opening jump (from $C_{i-1}$ to $O_i$),
and assume it is modeled
by a Poisson process.  
\end{Remar}
Thus for $i\in\{1, 2, \dots, n\}$ we have (see Figure \ref{fig0}):
\begin{equation}
\mbox{OPEN}(i)=O_i=S_{i-1},\;
\mbox{CLOSE}(i)=C_i=S_{i-f},\;
\mbox{HIGH}(i)=H_i,\;
\mbox{LOW}(i)=L_i,
\label{notation}
\end{equation}
where:
\begin{equation}
H_i=\sup_{t\in[i-1,i-f]}S_t,\;\;L_i=\inf_{t\in[i-1,i-f]}S_t.
\label{high_low}
\end{equation}
\begin{figure}[htb]
\includegraphics{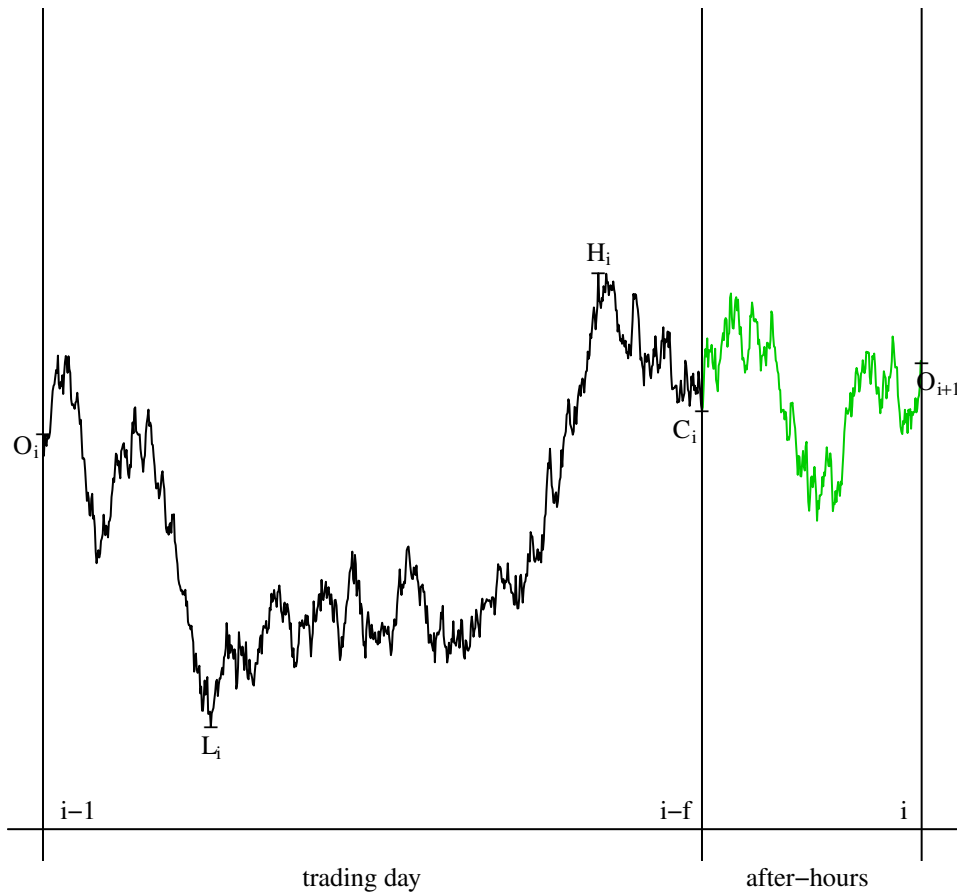}
\caption{A one-day period consisting of a trading day and an after-hours period}
\label{fig0}
\end{figure}
The evolution of the price during the trading period $i-1\leq t<i-f$ is given by:
\begin{equation}
log \;S_t=log\; O_i 
+\mu\;(t-i+1)+\sigma (W_t-W_{t-i+1}),
\label{trad}
\end{equation}
and during the after-hours virtual trading period $i-f\leq t<i$ by:
\begin{equation}
log \;S_t=log \;C_i
+\mu (t-i+f)+\sigma (W_t-W_{i-f}).
\label{virt}
\end{equation}


Taking expectation in (\ref{trad}) when $t\nearrow (i-f)$ 
and in (\ref{virt}) when $t\nearrow i$ gives:
\begin{eqnarray}
&&E\Big{[}\log \frac{C_i
}{O_i
}\Big{]}
=E\Big{[}\log\frac{S_{i-f}}{S_{i-1}}\Big{]}=\mu\;(1-f),\label{drift1}\\
&&E\Big{[}\log \frac{
O_{i+1}
}{C_i
}\Big{]}
=E\Big{[}\log\frac{S_{i}}{S_{i-f}}\Big{]}=\mu\;f.\label{drift2}
\end{eqnarray}
Using $W_t-W_s$ identically distributed to $W_{t-s}$, the trading day and 
virtual trading day variances are obtained, respectively, as:
\begin{alignat}{4}
\operatorname{VAR}&\Big{[}\log
\frac{C_i
}{O_i
}\Big{]}=
\operatorname{VAR}\Big{[}\log
\frac{S_{i-f}}{S_{i-1}}\Big{]}
=\;\sigma^2\;(1-f),\nonumber\\
\operatorname{VAR}&\Big{[}\log
\frac{O_{i+1}
}{C_i
}\Big{]}=
\operatorname{VAR}\Big{[}\log
\frac{S_{i}}{S_{i-f}}\Big{]}
=\;\sigma^2\;f.\label{variances}
\end{alignat}
Thus, we can write heuristically:
\[
\sigma^2=\operatorname{VAR}\Big{[}\log
\frac{C_{i}}{O_i}\Big{]}+
\operatorname{VAR}\Big{[}\log
\frac{O_{i+1}}{C_i}\Big{]}
=\operatorname{VAR(trading\;day)}+
\operatorname{VAR(after\;hours)}.
\]

To estimate the variance over the trading day we use the method of moments.
The range $R_{1-f}=\log H_1-\log L_1$ 
of the arithmetic Brownian motion $X_t=\log S_t$ over the trading day $[0,1-f]$
was obtained in equation (\ref{er}):
\begin{equation}
E(R_{1-f})=ER(\mu,\sigma,1-f)=ER(\mu(1-f),\sigma\sqrt{1-f},1).\label{new_er}
\end{equation} 
In (\ref{new_er}) we estimate $E(R_{1-f})$ using the daily range data:
\begin{equation}
k_1:=\frac 1n \sum_{i=1}^n \log \frac{H_i}{L_i}.
\label{equation9}
\end{equation}
 and  $\mu(1-f)$  by (see (\ref{drift1})): 
\begin{equation}
k_2:=\frac 1n \sum_{i=1}^n 
\log\frac{C_i
}{O_i
}. 
\label{equation8}
\end{equation}
This leads to the following equation to be solved for $x$, the estimate of 
$\sigma\sqrt{1-f}$:
\begin{equation}
k_1=h\Big{(}\frac{k_2}{x},\frac{x^2}{k_2}\Big{)}.
\label{main_eq}
\end{equation}
The squared of this solution gives an estimate $V_i=x^2$ of
the variance (volatility squared) corresponding to the trading day part of 
a one-day period. 

For the after-hours part of the one-day period we have two choices: $V_0$
(centered approach) used in Yang and Zhang (2000), or $V_0'$ (non-centered)
used in Garman and Klass (1980). Using the former (i.e. the sample standard variance $V_0$),
we obtain the estimate for the variance of the entire one-day period as:
\begin{equation}
V_Z:=V_0+V_i,\label{est}
\end{equation}
or, in annualized form, as:
\begin{equation}
\sigma_a^2:=252 V_Z.\label{ann}
\end{equation}


\medskip

\noindent Denoting by $V_C$ the sample variance of $\log(C_i/O_i)$ used in their estimator 
by Yang and Zhang (2000) :
$$V_{YZ}=V_0+kV_C+(1-k)V_{RS},$$ 
where $V_{RS}$ is the estimator of Rogers and Satchell (1991) and 
Rogers, Satchell and Yoon (1994) and $k$ is a constant, we note the following.
\begin{Remar}\label{comparison}
i) The term $V_i$ replaces the linear combination of $V_C$ and
$V_{RS}$ used by Yang and Zhang (2000) for the intra-day trading period, and it does not need estimating the value
of $k$ that achieves minimum variance. 

ii) Our estimator is a true range-based estimator (log-range to be precise
since it uses $\log({H_i}/{L_i})$), 
unlike that of Yang and Zhang (2000). 

iii) Our estimator $V_Z$ is independent of both the drift and the weight $f$ of
the after-hours information.
\end{Remar}

%

%
%

\begin{Examp}\label{examp30}
Consider the market data on the high, low, opening and closing prices for the IBM stock 
for the period from May 26, 2010 to June 18, 2010.
For each of these days we consider the historical 3-month \footnote{For parameter estimation Hull 
(2006, p. 287) recommends using historical data of 90 to 180 days.} estimates of $k_1$ and $k_2$, 
and we solve the corresponding equation (\ref{main_eq}). 

The solution is our estimate of the volatility $\sigma\sqrt{1-f}$ corresponding to the trading day, 
and we present it in annualized form (i.e. multiplied by $\sqrt{252}$) in Figure \ref{fig_chart0}.
We compare our estimate of the volatility corresponding to
a one-day period with the one of Yang and Zhang (2000). On June 18, 2010 they are
$\sigma_a=0.2781$ (see (\ref{ann})) and $0.2982$ (annualized
volatility corresponding to $V_{YZ}$).

\begin{figure}[htb]
\centerline{\epsfig{file=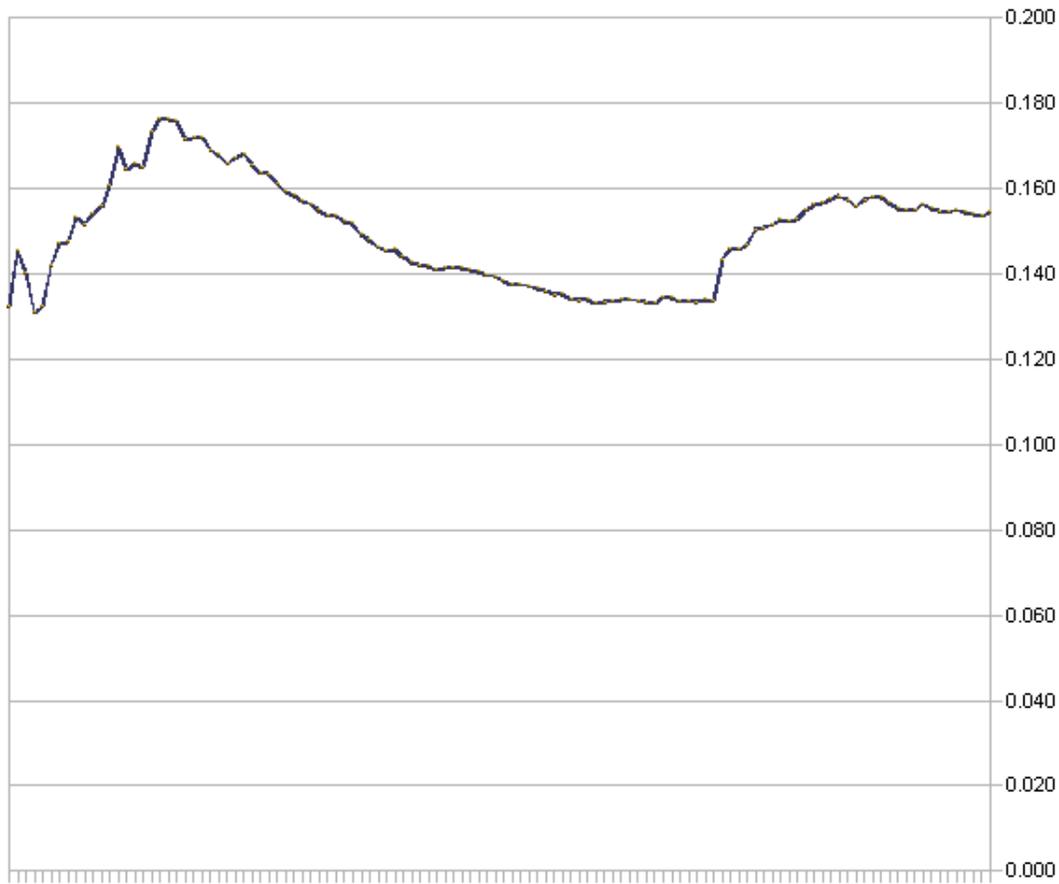,scale=0.7,angle=0}}
\caption{Estimated intra-day IBM volatility - May 26 to June 18, 2010 }
\label{fig_chart0}
\end{figure}



\end{Examp}






\section{European options: mispricing opportunities}


We use the resulting annualized volatility to compute the Black-Scholes prices of European options on the
stock. 
We then seek those instances when the computed prices differ the most from the 
market prices, and devise trading strategies to take advantage of the price
difference.

\medskip

We now use the volatility parameter estimated above to price European call 
options
using the Black-Scholes formula:
\begin{equation}
C_t=S_t\;\Phi(d_1)-Ke^{-r(T-t)}\;\Phi(d_2),                      \label{equation14}
\end{equation}
with
\begin{eqnarray*}
d_1&=&\frac{\log (S_t/K)+(r+\sigma_a^2/2)(T-t)}{\sigma_a\sqrt{T-t}},\\
d_2&=&d_1-\sigma_a\sqrt{T-t}.
\end{eqnarray*}

We devise a trading strategy to take advantage of the information differential
between our estimated prices and market prices. For simplicity we trade only in
European call options, and assume that at expiry there is a payment equal to
the payoff so that no actual trading occurs in the underlying stock (naked
trading).

Having assumed a constant volatility there is no volatility smile and no 
stochastic volatility\footnote{Alternative approaches like
stochastic volatility or econometric models (ARMA, GARCH etc) are not discussed
here.}, so we restrict our analysis to 
European call options 
whose strike prices are relatively close to the stock price at the beginning 
of the period (preferably in the money), and whose expiry dates are up to 
three months 
(the parameters can be re-estimated later in view of new data).
\begin{Examp}\label{examp2-0}
Consider the market prices for the European call options on IBM for the period
May 26, 2010 to June 18, 2010 with expiry dates June 18 and July 16,
and strike prices $K\in\{115, 120, 125, 130\}$ 
(the stock price on May 26 was 125.91).
We compare these market prices with the Black-Scholes prices calculated
using (\ref{equation14}). Here the inputs are the stock price, the volatility estimated in 
Example \ref{examp30}, and $r$ the value of the 1-month US Treasury 
bill yield for the previous day (online
\href{
http://www.treasury.gov/resource-center/data-chart-center/interest-rates/Pages/default.aspx
%http://www.treasury.gov/offices/domestic-finance/debt-management/interest-rate/yield_historical.shtml}
{Treasury data}\footnote
{http://www.treasury.gov/resource-center/data-chart-center/interest-rates/Pages/default.aspx}
\end{Examp}

Since  
the intra-day volatility of Example \ref{examp30} that we use in Black-Scholes formula
does not include the effect of the after-hours evolution,
we compensate by allowing our prices to differ by up to 10\% from the bid-ask spread. Thus, 
we trade when our estimated call price falls outside the interval (0.9$\times$bid-price, 1.1$\times$ask-price). 

There are two cases.
If our price is lower, then we short-sell the option at the bid price and wait for the first day when the 
estimated price is no longer lower to buy back the option at the then ask price. If it expires and the call 
option is exercised then we buy the stock in the market and deliver it.

If our price is higher, then we buy the option at ask price and wait for the first day when the price is no
longer higher to sell it at the then bid price. If it reaches expiry date, then we exercise it.

This trading strategy is summarized in Algorithm \ref{algo0} for $t$ between May 26, 2010 and
June 18, 2010 for European call options expiring at close June 18, 2010. The data is retrieved once a day,
except on expiration date when it is retrieved several times a day (this can be implemented as an algorithmic
trading strategy and deployed continuously without much effort, especially by those interested in technical
trading).

%
%
%
%
%
%
%
\begin{table}[htb]
\begin{center}
\begin{threeparttable}
\begin{tabular}{|c|c|c|c|c||c|c|c|c|}
\hline
 t&K&$\hat{C}(t)$ &(bid,ask)& trade&  t& (bid,ask) &trade& profit
\\
\hline
May 26&130&0.90&(\textcolor{blue}{1.16},{1.17})&sell&May 27&(0.96,\textcolor{red}{0.99})&buy&0.17
\\
May 28&130&0.57&(\textcolor{blue}{0.75},0.78)&sell&Jun 2&(0.67,\textcolor{red}{0.70})&buy&0.05
\\
Jun 7&125&1.92&(\textcolor{blue}{2.20},2.23)&sell&Jun 8&(1.21,\textcolor{red}{1.23})&buy&0.97
\\
Jun 7&130&0.29&(\textcolor{blue}{0.42},0.44)&sell&Jun 8&(0.15,\textcolor{red}{0.17})&buy&0.25
\\
Jun 8&120&3.69&(\textcolor{blue}{4.15},4.20)&sell&Jun 9&(4.60,\textcolor{red}{4.75})&buy&(0.6)
\\
Jun 9&125&1.09&(\textcolor{blue}{1.30},1.38)&sell&Jun 10&(3.05,\textcolor{red}{3.15})&buy&(1.77)
\\
Jun 17&130&1.09&(0.90,\textcolor{red}{0.94})&buy&Jun 19&(\textcolor{blue}{1.00},1.05)&sell&0.06
\\
Jun 18\tnote{a} &130&0.60&(0.48,\textcolor{red}{0.51})&buy&Jun 18\tnote{b} &(\textcolor{blue}{0.63},0.69)&sell&0.12
\\
Jun 18\tnote{c} &130&0.18&(\textcolor{blue}{0.21},0.25)&sell&Jun 18\tnote{d} &$S_t$=130.14&ex\tnote{e}&0.07
\\
\hline
\end{tabular}
\begin{tablenotes}[para]
\item [a] at 12:27pm 
\item  [b] at 13:36pm
\item [c] at 15:58pm
\item [d] at 16:00pm
\item [e] if exercised
\end{tablenotes}
\end{threeparttable}
\caption{Trading in the call option expiring June 18, 2010 \newline
\centerline{(left: open a position, right: close position)}}
\label{table10}
\end{center}
\end{table}
\begin{Remar}\label{algo}
This strategy results in an overall loss of 0.68 (see Table \ref{table10}). This is due mostly 
to one large loss induced by a large sudden move in the stock price
on June 10 (127.3 versus 123.9 the day before). That is because we use yesterday's
intra-day volatility to trade in today's world. 

Over a time horizon longer than a month the strategy can absorb
such shocks in the stock prices, provided they are sparse. Alternatively, one can
implement an additional stopping rule when the change in the stock price exceeds a pre-determined margin.

A similar behaviour is exhibited when applying the same trading strategy to
the European call option expiring at close July 16, 2010, but as the expiry date
is longer than a couple of months the limitations of the assumptions of the model become apparent.
\end{Remar}

\section{Conclusions}
We have used the method of moments to estimate the volatility of the stock price
and used this to identify arbitrage opportunities in the market of European
options. As a by-product we have derived the density and expectation of the
range of an arithmetic Brownian motion.

In comparison to the estimate of Yang and Zhang (2000), our volatility estimate takes
advantage of the actual range of the Brownian motion and perhaps does not
overestimate as much. It is most useful for short expiration dates and for strike prices
that are not far out. We believe it is an efficient alternative that
can be easily computed and has a practical implementation. These traits recommend it
to the attention of practioners in the field.

\newpage

\hspace*{0.6cm}{\bf REFERENCES}

\bigskip

\noindent
\textsc{ Billingsley, P.} (1968):
{\em Convergence of probability measures.} John Wiley, New York.

\noindent
\textsc{ Chan, L. and D. Lien} (2003):
Using high, low, open, and closing prices to estimate the effects of cash
settlements on futures prices.
{\em International Review of Financial Analysis\/} 12, 35--47.



\noindent
\textsc{ Fama, E.F.} (1965):
The behaviour of stock market prices.
{\em Journal of Business\/} 38, 34--105.

\noindent
\textsc{ Feller, W.} (1951):
The asymptotic distribution of the range of sums of independent random 
variables.
{\em Annals of Mathematical Statistics\/} 22, 427--432.



\noindent
\textsc{ Garman, M. and M. Klass } (1980):
On the estimation of security price volatilities from historical data.
{\em Journal of Business\/} 53(1), 67--78.


\noindent
\textsc{ Harrison, J.M.} (1985):
{\em Brownian motion and stochastic flow systems.}
Wiley, New York.

\noindent
\textsc{ Hull, J.C.} (2006):
{\em Options, futures and other derivatives.} 6th ed,
Prentice Hall, New Jersey.

\noindent
\textsc{ Karatzas, I. and S. E. Shreve} (1998):
{\em Brownian motion and stochastic calculus.}
Springer, New York.





\noindent
\textsc{ Kon\'{e}, F. J.} (1996):
{\em Estimation of the volatility of stocks using the high, low and 
closing prices.}
Ph.D. thesis, State University of New York at Stony Brook.

\noindent
\textsc{ L\'{e}vy, P. } (1948):
{\em Processus stochastique et mouvement brownien.} Gauthier-Villars, Paris. 

\noindent
\textsc{ Magdon-Ismail, M. and A. Atiya } (2000):
Volatility estimation using high, low and close data - a maximum likelihood
approach.
{\em Computational Finance\/}, June CF2000 Proceedings.

\noindent
\textsc{ Magdon-Ismail, M., A. Atiya, A. Pratap and Y. Abu-Mustafa } (2004):
On the maximum drawdown of a Brownian motion.
{\em Journal of Applied Probability\/}, 41(1), 147--161.


\noindent
\textsc{ Parkinson, M.} (1980):
The extreme value method for estimating the variance of the rate of return.
{\em Journal of Business} 53(1), 61--65.

\noindent
\textsc{ Rogers, L.C.G. and S. Satchell } (1991):
{ Estimating variance from high, low and closing prices}. 
{\em Annals of Applied Probability\/} 1(4), 504--512.

\noindent
\textsc{ Rogers, L.C.G., S. Satchell and Y. Yoon} (1994):
Estimating the volatility of stock prices: a comparison of methods that use
high and low prices.
{\em Applied Financial Economics\/} 4, 241--247.

\noindent
\textsc{ Rogers, L.C.G. and F. Zhou} (2008):
Estimating correlation from high, low, opening and closing prices.
{\em The Annals of Applied Probabilty\/} 18(2), 813--823.

\noindent
\textsc{ Sutrick K., J. Teall, A. Tucker and J. Wei} (1997):
{ The range of  Brownian motion processes: density functions and derivative
pricing applications}.
{\em The Journal of Financial Engineering} 6(1), 31--46.

\noindent
\textsc{ Yang D. and Q. Zhang} (2000):
{  Drift-independent volatility estimation based on high,
low, open, and close prices}.
{\em Journal of Business} 73, 477--491.





\IncMargin{1em}
\begin{algorithm}
 \SetKwFunction{BlackScholesCall}{BlackScholesCall}
 \SetKwFunction{profit}{profit}
 \SetKwInOut{Input}{input}
 \SetKwInOut{Output}{output}
\BlankLine
\Input{Parameters $t$, $\sigma_a$, $r$, $S_t$, bid(t) and ask(t)  (European call prices)}
\Output{Profit of trading strategy}
\BlankLine
\profit $\leftarrow$ 0 \tcp*{initialize}
T $\leftarrow$ June 18, 2010\tcp*{expiry date}
  compute $\sigma_a$\tcp*{Example \ref{examp30}}
  \For
  {
    \emph{K} $\leftarrow$ \emph{115} \KwTo \emph{130}
  }
  { 
     $\hat{C}(t)\leftarrow$ \BlackScholesCall{$t, T, K, \sigma_a, r, S_t$}\tcp*{(\ref{equation14})}
     \eIf
     {
       $\hat{C}<0.9\times$ \emph{bid(t)}
     }
     {
      \profit $\leftarrow$ \profit + bid(t)\tcp*{short-sell call}
      \While{$t<T$ \emph{and} $\hat{C}$\emph{(t)}$<$\emph{0.9}$\times$\emph{bid(t)}}{t $\leftarrow$ t+1
                                                   \tcp*{wait 1 day}
          compute $\sigma_a$\tcp*{Example \ref{examp30}}
        $\hat{C}(t)\leftarrow$ \BlackScholesCall{$t, T, K, \sigma_a, r, S_t$}\tcp{(\ref{equation14})}
         }
    \eIf{$t<T$}{\profit $\leftarrow$ \profit - ask(t) \tcp*{buy back call}}
      {\If{\emph{call is exercised}}{\profit $\leftarrow$ \profit - ($S_t$ -K)\;\tcp{buy stock and deliver for K}
         }}
      }
      {
           \If
          {  
             $\hat{C}>1.1\times$\emph{ask(t)}
           }
        {
          \profit $\leftarrow$ \profit - ask(t)\tcp*{buy call}
          \While{$t<T$ \emph{and} $\hat{C}$\emph{(t)}$>$\emph{1.1}$\times$\emph{ask(t)}}
         {
          t $\leftarrow$ t+1\tcp*{wait 1 day}
          compute $\sigma_a$\tcp*{Example \ref{examp30}}
           $\hat{C}(t)\leftarrow$ \BlackScholesCall{$t, T, K, \sigma_a, r, S_t$}\;
         }
    \eIf{$t<T$}{\profit $\leftarrow$ \profit + bid(t) \tcp*{sell call}}
         { \profit $\leftarrow$ \profit + ($S_t$ -K)\tcp*{exercise call}
         }}    
    } 
} 
\Return{\profit}\;
\BlankLine
\caption{Trading strategy for a mispricing opportunity found at time $t$ 
}\label{algo0}
\end{algorithm}\DecMargin{1em}
\end{document}